\documentclass{article}
\usepackage[utf8]{inputenc}
\usepackage{authblk}
\usepackage{setspace}
\usepackage[margin=1.25in]{geometry}
\usepackage{graphicx}
\usepackage{subcaption}
\usepackage{amsmath}
\usepackage{amssymb}
\usepackage{amsthm}
\usepackage{bm}
\usepackage{hyperref}
\hypersetup{colorlinks=true,linkcolor=blue,citecolor=blue,urlcolor=blue}
\usepackage{xcolor}
\usepackage{booktabs}
\usepackage{float}
\usepackage{enumitem}
\setlist{itemsep=3pt,parsep=1pt}
\usepackage{microtype}
\usepackage{dblfloatfix}
\usepackage[section]{placeins}
\usepackage{tabularx}
\usepackage{array}
\usepackage{lineno}

\newcolumntype{Y}{>{\raggedright\arraybackslash}X}

\newtheorem{theorem}{Theorem}
\newtheorem{proposition}{Proposition}
\newtheorem{definition}{Definition}
\newtheorem{corollary}{Corollary}
\newtheorem{remark}{Remark}
\newtheorem{assumption}{Assumption}

\usepackage[style=phys,
citestyle=numeric-comp,
sorting=none,
doi=true,
url=false,
isbn=false,
natbib=true]{biblatex}
\title{Trajectory Stability and Signature Diagnostics for Comet-Based Interstellar Navigation}

\author[1]{Bo Pieter Johannes Andr\'ee\thanks{The findings, interpretations, and conclusions expressed herein are entirely those of the author and do not necessarily represent the views of the International Bank for Reconstruction and Development/World Bank, its Board of Executive Directors, or the governments they represent. No World Bank resources were used to conduct this research. This research received no specific funding. Domain-specialist input is gratefully welcomed. Contact: bandree(at)worldbank.org}}

\affil[1]{Data Group and Multilateral \& UN, World Bank, Geneva, Switzerland.}
\date{March 2026}

\begin{document}

\maketitle

\begin{abstract}
Interstellar objects (ISOs) motivate a coupled mission-design and inference question relevant to spacecraft dynamics and control in extreme environments: if volatile-rich, rotating comet-like bodies were used for sustained deep-space navigation by exploiting pre-existing hyperbolic motion and in-situ propellant, what stability requirements arise under non-gravitational forcing, and what astrometric signatures might distinguish active stabilization from uncontrolled natural dynamics? We develop a stability-theoretic framework for trajectory tracking with jet-actuated correction, and show that high-speed transit geometry---including debris-belt avoidance and encounter phasing---tightly constrains feasible trajectories, making long-horizon tracking stability mission-critical. We model tracking residuals as the balance of disturbances and corrective action, and derive stability conditions across four levels: disturbance-energy stability, outer-loop contraction, actuator-memory stability, and rotation-mediated (Floquet) stability. The analysis implies residual diagnostics that can motivate empirical tests: under comparable forcing, effective stabilization is expected to strengthen short-horizon error correction, reduce event-conditioned persistence and variance clustering, regularize standardized innovations, and yield bounded post-shock recovery. More broadly, the framework provides a reference for deep-space guidance and control under nonlinear, multi-field disturbances and for planetary-defense concepts involving attitude shaping or impulsive kinetic impact.
\end{abstract}


\pagebreak


\section{Introduction}
\label{sec:intro}

Interstellar objects (ISOs) and highly eccentric comet-like bodies periodically traverse the inner Solar System, offering rare but information-rich opportunities for remote characterization and in-situ exploration. The confirmed detections of 1I/\hspace{0pt}'Oumuamua~\cite{meech2017oumuamua}, 2I/Borisov~\cite{jewitt2019comet}, and 3I/ATLAS~\cite{loeb2025intercepting} demonstrate that such bodies pass through the inner Solar System with sufficient regularity~\cite{delaFuenteMarcos2024ejected} to support systematic characterization~\cite{eldadi2026ioss,jewitt2022interstellar,delaFuenteMarcos2025gtc} and mission planning~\cite{hein2022interstellar,seligman2018oumuamua,hibberd2020project,landau2023iso}. These developments motivate renewed attention to the spacecraft dynamics and control challenges inherent in exploiting such bodies for exploration~\cite{zhang2024advances_sdc}.

Beyond their scientific interest as remnants of planetary formation, ISOs can be viewed through an engineering lens: their high-energy heliocentric trajectories---with hyperbolic excess velocities of 20--80~km/s~\cite{jewitt2022interstellar}---provide pre-existing large-baseline displacement that could, in principle, be leveraged for deep-space navigation if trajectory control is achievable. Practical steering on a comet-like body may be accomplished with hardware that enables a combination of instantaneous in-plane authority and attitude shaping, exploiting natural rotation to achieve six-degree-of-freedom control authority in a static force sense.

A central challenge is that controlled navigation on a small rotating body is an inherently \emph{stochastic} tracking problem. Even when the desired trajectory is smooth, realized motion is continuously perturbed by solar radiation pressure~\cite{wertz1999spacecraft}, plasma and particle environments~\cite{jokipii1998particle,richardson2018solar}, episodic space-weather events~\cite{tsurutani2003extreme}, micro-impacts, and (for volatile-rich bodies) variable outgassing~\cite{mumma2011composition}. It is not enough to produce occasional thrust; the control system must ensure \emph{long-horizon stability}: tracking errors must remain bounded and must ``forget'' early transients so that initial errors do not dominate future trajectory planning. This requirement---stable trajectory tracking under persistent stochastic perturbations in a multi-field coupling, nonlinear environment---places the problem squarely within the domain of spacecraft dynamics and control~\cite{wie2008spacecraft,markley2014fundamentals,zhang2024advances_sdc}.

The stability framework developed here is grounded in geometric trajectory constraints. For interstellar transit involving inner-system reconnaissance, the trajectory must satisfy debris-belt avoidance on both entry and exit, achieve favorable alignment with priority targets, and maintain propellant economy. These constraints imply an asymmetry in propulsion requirements: steering demands are modest (degree-scale slope adjustments executed over weeks to months), while stability demands are critical (maintaining the designed trajectory within tight tolerances over mission-critical transit points). Once a trajectory profile is selected, the dominant operational requirement is not ad hoc maneuvering authority but the ability to remain on course under persistent stochastic perturbations.

We develop a multi-level stability theory for jet-actuated spacecraft maintaining such trajectories. The theory applies to any configuration that enables practical three-dimensional steering, from minimal designs~\cite{andree2026thruster} to over-actuated systems with full six-degree-of-freedom authority. Closed-loop tracking is stable only if disturbance energy remains bounded, tracking dynamics are contractive, actuation dynamics do not self-excite, and rotation-mediated time-periodic authority remains Floquet-stable~\cite{bittanti2009periodic}. This multi-level structure resonates with broader challenges in spacecraft attitude and orbit control where coupled dynamics, extreme environments, and limited actuation interact~\cite{zhang2024advances_sdc,li2024trajectory_mars}.

The stability requirements also yield observable predictions. The ISO literature commonly frames anomaly detection by stacking individual deviations from a gravity-only trajectory~\cite{eldadi2026ioss}, but this comparison is a poor discriminator in both directions: natural bodies routinely depart from gravitational paths through asymmetric outgassing and radiation pressure~\cite{micheli2018oumuamua,bockelee2004borrelly,jockers2011encke}, while an actively steered craft could follow a near-gravitational trajectory precisely because doing so is fuel-efficient. Each individual deviation admits a passive explanation, so the stacking approach is vulnerable to model expansion on the passive side. We develop tests that instead target a structural property: passive dynamics, however complex, do not produce closed-loop contraction toward a reference trajectory. The diagnostics we derive are designed to detect violation of this structural absence rather than to accumulate quantitative anomalies against an expanding natural baseline.

The paper proceeds from trajectory geometry to a gravity-referenced stochastic tracking formulation, derives multi-level stability conditions, and develops residual-based diagnostics with a numerical illustration connecting stability requirements to observable signatures. We conclude with the observational and mission-design work ahead: discriminating evidence in future ISO encounters, and leveraging comet-based platforms for deep-space exploration.


\section{Trajectory geometry: debris avoidance and reconnaissance constraints}
\label{sec:trajectory_geometry}

Before addressing the stability problem, we establish the geometric constraints that govern trajectory design for high-speed interstellar transit with multi-target reconnaissance objectives. A reconnaissance pass through an inner planetary system requires safe entry through any outer debris structure, close-to-ecliptic passage near priority targets to minimize encounter range and maximize observation time, and safe exit.

Planetary formation generically leaves behind a belt of debris---unaccreted planetesimals, dust, and small bodies---concentrated near the ecliptic. Among FGK-type host stars, such debris structures extend to tens of AU while habitable zones fall at $\sim$0.5--2~AU~\cite{wyatt2008debris_disks,hughes2018debris}, so our solar system is representative of the target population. We therefore develop a typical mission geometry using it as the reference configuration.

At velocities approaching $80\,\mathrm{km/s}$, an inner planetary system comparable to our own is crossed in weeks, and large lateral $\Delta V$ for close planetary flybys may be prohibitive. In this regime, transit geometry materially constrains encounter distances and the entry and exit points through the debris belt, and trajectory slope selection becomes a first-order mission parameter. These constraints are characteristic of trajectory optimization in deep-space exploration more broadly~\cite{li2024trajectory_mars,colombo2009lowthrust_neo}.

We first formalize the belt-crossing slope constraint and then relate slope and node placement to encounter range through the vertical-offset term. A key implication is that belt-avoidance, node placement, encounter timing, and instrument requirements interact in ways that cannot be readily re-optimized during flight, reinforcing that long-horizon tracking stability is the dominant operational requirement once a profile is selected.

\subsection{Debris avoidance geometry}
\label{sec:debris_geometry}

Proposition~\ref{prop:debris_slope} gives the minimum geometric requirement: a trajectory that enters a debris belt on one side and exits on the other must exceed a slope threshold set by the belt aspect ratio $h/d$. This is a locally necessary condition; it does not guarantee avoidance everywhere along the chord.

\begin{proposition}[Minimum avoidance slope for opposite-side belt crossings]
\label{prop:debris_slope}
Model a debris belt as concentrated about a reference plane with effective full vertical thickness $h$ (dense region $|z|\le h/2$) at the belt crossing locations. Consider an inbound belt crossing and an outbound belt crossing separated by in-plane chord length $d$ along the trajectory segment connecting the two crossings (e.g., of order a belt diameter). A straight-line trajectory with constant slope $\theta$ relative to the reference plane can place the two crossings on opposite sides of the dense band (one at or above $+h/2$ and the other at or below $-h/2$, or vice versa) only if
\begin{equation}
\theta \;\ge\; \theta^* \;=\; \arctan\!\left(\frac{h}{d}\right).
\label{eq:min_slope}
\end{equation}
\end{proposition}

\begin{proof}
Over the chord length $d$, a line at slope $\theta$ changes altitude by $\Delta z = d\tan\theta$. Opposite-side crossings require $|\Delta z|\ge h$, with equality at the minimum. Hence $d\tan\theta^* = h$, yielding \eqref{eq:min_slope}.
\end{proof}

\begin{remark}[Scope and margin]
Proposition~\ref{prop:debris_slope} is a crossing constraint. It is not a collision-risk model and does not claim that the trajectory avoids all belt material everywhere along the chord. Operational planning would typically introduce explicit margin (replace $h$ by $h+m$) to account for uncertainty in belt thickness and structure, and would complement the geometry with a debris-density and shielding assessment.
\end{remark}

For the Kuiper Belt, taking an effective thickness $h \approx 10\,\mathrm{AU}$ and chord length $d \approx 100\,\mathrm{AU}$~\cite{petit2011canada,bannister2018ossos},
\begin{equation}
\theta^*_{\mathrm{KB}} = \arctan(0.1) \approx 6^\circ,
\label{eq:kuiper_slope}
\end{equation}
with plausible variations giving $\theta^*\approx 5^\circ$--$7^\circ$.

\subsection{Encounter range and node placement}
\label{sec:node_placement}

Having established the debris-avoidance constraint, we turn to how the chosen slope affects encounter range with interior targets. Geometrically, range decomposes into an in-plane term and a vertical term. Moreover, entering above the belt and exiting below it forces the trajectory to cross the reference plane somewhere between the two belt crossings, creating a natural node along the transit line where the vertical offset vanishes. This makes node placement a practical design lever: one can locate the node near a priority target, after which the remaining question is when (and for which encounters) reducing slope yields a meaningful additional range reduction.

\begin{proposition}[Encounter range decomposition]
\label{prop:range}
Fix a target body and consider the instant the spacecraft reaches the target's heliocentric radius $r$. Let $s$ denote the in-plane separation between the target and the spacecraft's in-plane crossing point at that instant, and let $v$ denote the spacecraft's perpendicular distance to the reference plane. Then the instantaneous range satisfies
\begin{equation}
R \;=\; \sqrt{s^2 + v^2}.
\label{eq:R}
\end{equation}
\end{proposition}

\begin{proof}
At the encounter epoch, decompose the relative displacement into orthogonal components: an in-plane component of magnitude $s$ and a perpendicular-to-plane component of magnitude $v$. The Pythagorean theorem gives \eqref{eq:R}.
\end{proof}

Equation~\eqref{eq:R} shows that slope can affect range only through the vertical term $v$; it cannot compensate for poor in-plane alignment (large $s$). To make this dependence explicit along the transit line, we parameterize $v$ via a node location.

\begin{proposition}[Node representation of vertical offset]
\label{prop:node}
Parameterize the straight-line transit segment between the two belt crossings by along-track distance $x\in[0,d]$ measured in the reference plane from the inbound crossing ($x=0$) to the outbound crossing ($x=d$). If the trajectory has constant slope $\theta$, then its vertical coordinate is affine:
\begin{equation}
z(x) \;=\; (x_0-x)\tan\theta,
\label{eq:zx}
\end{equation}
for a unique $x_0\in\mathbb{R}$ at which the trajectory intersects the reference plane ($z(x_0)=0$). The vertical offset is $v(x)=|z(x)|$.
\end{proposition}

\begin{proof}
A straight line in the $(x,z)$ plane has the form $z(x)=a-bx$ with $b=\tan\theta$. The point $x_0=a/b$ satisfies $z(x_0)=0$, giving \eqref{eq:zx}.
\end{proof}

\begin{remark}[Consequence of entering above and leaving below]
If the inbound crossing is above the plane and the outbound crossing is below it (or vice versa), then $z(0)$ and $z(d)$ have opposite signs. By continuity, there exists at least one $x_0\in(0,d)$ with $z(x_0)=0$. Thus, for any opposite-side belt-crossing design, there is a node along the interior transit line at which vertical offset vanishes.
\end{remark}

Corollary~\ref{cor:meaningful} formalizes when reducing $\theta$ yields a meaningful reduction in $R$ for a given encounter.

\begin{corollary}[When slope reduction yields meaningful range reduction]
\label{cor:meaningful}
Consider a given target encounter occurring at some $x=x_\star$. With $R$ as in \eqref{eq:R} and $v(x_\star)=|x_0-x_\star|\tan\theta$ from \eqref{eq:zx}, the marginal value of reducing $\theta$ is large only when
\[
s \;\lesssim\; v(x_\star)
\qquad\text{and}\qquad
|x_0-x_\star| \text{ is not small.}
\]
If $s\gg v(x_\star)$ (poor in-plane alignment) then $R$ is dominated by $s$ and slope reduction has weak effect. If $x_\star\approx x_0$ (the encounter is near the node) then $v(x_\star)\approx 0$ and slope is not a binding range constraint for that encounter.
\end{corollary}

\begin{proof}
From \eqref{eq:R}, $\partial R/\partial v = v/R$. Thus sensitivity to slope is suppressed when $v\ll R$ (equivalently $s\gg v$), and it vanishes when $v=0$, i.e., at the node.
\end{proof}

\begin{remark}[Node placement and the reconnaissance template]
\label{rem:recon_template}
Because a node necessarily exists for opposite-side belt crossings, mission design is primarily about where to place it. For a reconnaissance mission targeting inner-system planets, the habitable zone ($\sim$1--1.5~AU) is the natural anchor: it concentrates the highest-priority targets, and placing the node there eliminates $v$ for the most valuable encounters. With the node so fixed and the slope set to $5^\circ$--$7^\circ$ by debris avoidance, the geometric profile is largely determined. The principal remaining freedom is temporal: encounter range still depends on the in-plane separation $s$ (Proposition~\ref{prop:range}), which is governed by the orbital phases of the target planets at the node-crossing epoch. In practice, mission design reduces to selecting a launch window that places the node crossing when the nearest habitable-zone planets are favorably aligned---minimizing the aggregate in-plane distance to the two or three closest targets. Slope reduction below the debris-avoidance minimum can further improve encounters far from the node (Corollary~\ref{cor:meaningful}), but this is a secondary refinement once the node and timing are set.
\end{remark}

\subsection{Illustrative geometry}
\label{sec:figure_geometry}

The preceding constraints converge on a narrow family of trajectory profiles: debris avoidance fixes the slope range, the habitable zone anchors the node, and launch-epoch selection targets favorable planetary phasing at the node-crossing time. Within this family, the slope varies as a \emph{profile} rather than remaining fixed. Consider the following instance:
\begin{enumerate}
\item[\textup{(i)}] \textbf{Entry:} $\theta = 7^\circ$, providing a $1^\circ$ margin above a $\sim 6^\circ$ Kuiper-belt scale avoidance slope during inbound debris transit.
\item[\textup{(ii)}] \textbf{Inner system:} after clearing the debris belt ($r \lesssim 30\,\mathrm{AU}$), reduce slope to $\theta = 5^\circ$ if this improves priority-target range and/or enlarges the set of epochs yielding acceptable alignment, or favors alignment-sensitive measurements.
\item[\textup{(iii)}] \textbf{Exit:} restore $\theta = 7^\circ$ before outbound debris transit.
\end{enumerate}
This requires a total steering budget of $4^\circ$ ($2^\circ$ down, $2^\circ$ up), executed over the months-long inner-system passage. The specific slope values and node location are mission-dependent, but the template itself---degree-scale slopes, habitable-zone node, phasing-driven timing---is generic. Once a profile is selected, the maneuver magnitudes are known and modest; the governing challenge is maintaining the profile under persistent stochastic perturbations.

Figure~\ref{fig:kuiper_geometry} visualizes the associated geometry tradeoffs. Panel~(a) shows an inner-system side view ($x$--$z$) comparing two slope choices ($7^\circ$ and $5^\circ$) for a habitable-zone node near $\sim 1\,\mathrm{AU}$ (Remark~\ref{rem:recon_template}), together with a Sun-centered reference case. Panel~(b) plots the corresponding vertical-offset magnitude $|z(r)|$ versus heliocentric distance, highlighting the separation in $|z|$ at Kuiper-belt scale and an illustrative outer-system corridor ($10$--$45\,\mathrm{AU}$) over which degree-scale trajectory shaping may be executed.

\begin{figure}[h!]
    \centering
    \includegraphics[width=\textwidth]{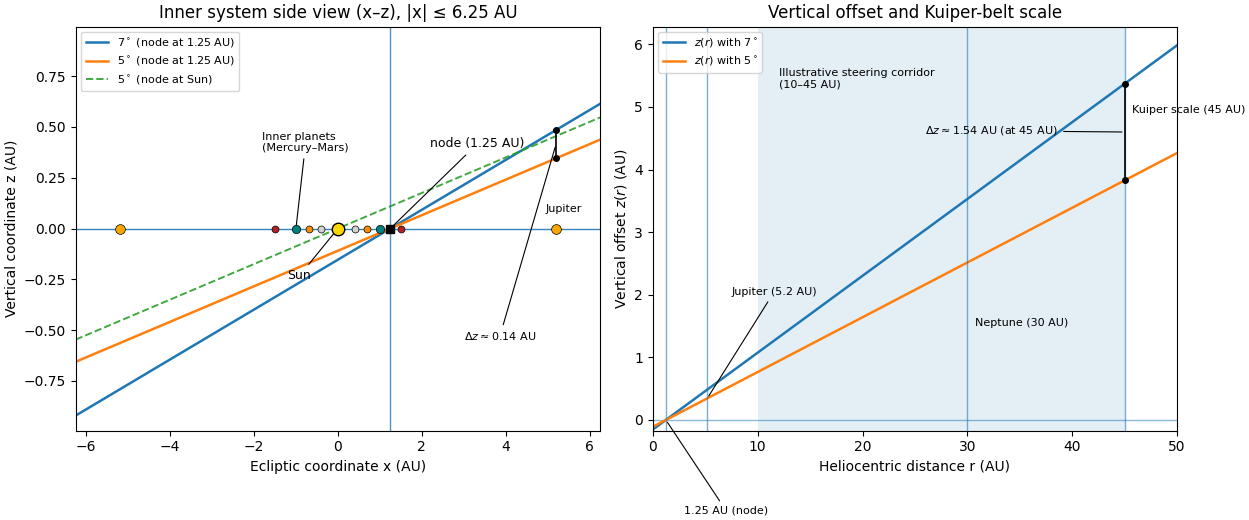}
    \caption{\textbf{Schematic geometry linking debris-belt avoidance, node placement, and inner-system vertical offset.}
    \textbf{(a)} Inner-system side view ($x$--$z$, $|x|\le 6\,\mathrm{AU}$). Solid lines compare two slope choices ($7^\circ$ and $5^\circ$) with the node placed near $1.25\,\mathrm{AU}$---a habitable-zone anchor that jointly minimizes encounter range to Earth and Mars when their orbital phases are favorable (Remark~\ref{rem:recon_template}); the dashed line shows a Sun-centered node for reference. The bracket at Jupiter indicates the reduction in vertical offset from lowering the inner-system slope.
    \textbf{(b)} Vertical offset magnitude $|z(r)|$ versus heliocentric distance for the same two slopes (node at $\sim 1.25\,\mathrm{AU}$). The shaded band marks an illustrative outer-system corridor ($10$--$45\,\mathrm{AU}$) for degree-scale trajectory shaping, and the bracket at $45\,\mathrm{AU}$ highlights the resulting separation in $|z|$ at Kuiper-belt scale.}
    \label{fig:kuiper_geometry}
\end{figure}

\section{Trajectory control and the stability problem}
\label{sec:problem}

\subsection{Stability over steering}
\label{sec:stability_over_steering}

The trajectory template established in Section~\ref{sec:trajectory_geometry} reveals an asymmetry in propulsion requirements that motivates the remainder of the paper. The steering budget is modest: the $4^\circ$ profile adjustment requires only degree-scale trajectory modification, executed over weeks to months, with no need for sharp turns or high angular rates. The stability requirement is far more demanding: the belt-crossing margin, node placement, and encounter phasing must all be preserved under persistent, stochastic perturbations---outgassing, radiation pressure, micrometeorite impacts, and thruster imbalances---that accumulate throughout the transit. A degree-scale drift from the planned belt-safe slope would erode the avoidance margin; a timing error at the node-crossing epoch would degrade the planetary alignment that the launch window was chosen to exploit. The governing problem is therefore not episodic steering authority but long-horizon tracking stability: errors must remain bounded and must decay so that early transients do not dominate later mission geometry. The remainder of the paper treats the planned profile as given and develops stability conditions and observable residual implications for maintaining it.

\subsection{Observed trajectories}

We consider a body whose heliocentric motion is observed through astrometry and whose short-term state can be estimated via nonlinear filtering. Let $\mathbf{x}_t\in\mathbb{R}^3$ and $\dot{\mathbf{x}}_t\in\mathbb{R}^3$ denote inertial position and velocity at discrete times $t=1,2,\dots$ with sampling interval $\Delta t$. Let $\mathbf{x}^{\mathrm{grav}}_t$ denote a gravity-only reference trajectory obtained by propagating an estimated initial state under a chosen gravitational dynamics model. This gravity reference provides a common baseline for both passive and controlled bodies, and it emphasizes a central distinction: gravitational forces are predictable over long horizons, whereas many non-gravitational effects are not.

We represent the realized trajectory using a decomposition that separates (i) gravity-dominated motion, (ii) long-horizon guidance, and (iii) residual deviations driven by stochastic forcing and correction:
\begin{equation}
\mathbf{x}_t
=
\mathbf{x}^{\mathrm{grav}}_t
+
\mathbf{x}^{\mathrm{plan}}_t
+
\mathbf{s}_t
-
\mathbf{o}_t .
\label{eq:master_decomp}
\end{equation}
Here $\mathbf{x}^{\mathrm{plan}}_t$ is a mission-level guidance component expressed \emph{relative} to the gravity-only baseline---incorporating the debris-avoidance and reconnaissance geometry developed in Section~\ref{sec:trajectory_geometry}---$\mathbf{s}_t$ aggregates non-gravitational forcing that is uncertain or only partially modeled, and $\mathbf{o}_t$ is the offset produced by short-horizon corrective authority. This separation captures the practical logic of navigation on ISO-class trajectories: long-horizon steering can be planned around predictable forces (gravity and any stable mean non-gravitational terms), while short-horizon offsetting must continuously counter unpredictable (or unavoidable) fluctuations and shocks.

Equation~\eqref{eq:master_decomp} is written to accommodate both passive comets and actively controlled bodies as special cases. For a natural comet, there is no mission guidance and no corrective offset, so
\begin{equation}
\mathbf{x}^{\mathrm{plan}}_t \equiv \mathbf{0},\qquad \mathbf{o}_t \equiv \mathbf{0}
\quad\Longrightarrow\quad
\mathbf{x}_t = \mathbf{x}^{\mathrm{grav}}_t + \mathbf{s}_t.
\label{eq:natural_comet}
\end{equation}
For an actively steered craft, $\mathbf{x}^{\mathrm{plan}}_t$ may be nonzero and $\mathbf{o}_t$ acts to compensate stochastic forcing:
\begin{equation}
\begin{aligned}
\mathbf{x}^{\mathrm{plan}}_t &\not\equiv \mathbf{0},\qquad \mathbf{o}_t \not\equiv \mathbf{0}
\quad\Longrightarrow\quad
\mathbf{x}_t
= \mathbf{x}^{\mathrm{grav}}_t + \mathbf{x}^{\mathrm{plan}}_t \\
&\hspace{3.6em} + \mathbf{s}_t - \mathbf{o}_t.
\end{aligned}
\label{eq:controlled_craft}
\end{equation}

This unified formulation clarifies what different empirical tests can and cannot identify: some procedures probe whether $\mathbf{x}_t$ is adequately explained by gravity plus plausible non-gravitational forcing, while others seek evidence that $\mathbf{o}_t$ is nonzero through signatures of active stabilization.

\subsection{Tracking error}

To connect \eqref{eq:master_decomp} to stability theory, define the gravity-referenced planned trajectory
\begin{equation}
\mathbf{x}^{\mathrm{ref}}_t \;\equiv\; \mathbf{x}^{\mathrm{grav}}_t + \mathbf{x}^{\mathrm{plan}}_t,
\label{eq:x_ref_def}
\end{equation}
and the corresponding tracking error
\begin{equation}
\mathbf{e}_t \;\equiv\; \mathbf{x}_t - \mathbf{x}^{\mathrm{ref}}_t.
\label{eq:e_def}
\end{equation}
Combining \eqref{eq:master_decomp}--\eqref{eq:e_def} yields the central residual identity
\begin{equation}
\mathbf{e}_t \;=\; \mathbf{s}_t \;-\; \mathbf{o}_t.
\label{eq:error_shock_offset}
\end{equation}

Let $\mathcal{F}_t$ denote the filtration generated by the observation process and all realized control actions
up to time $t$ (including any internal controller state). Stability of trajectory tracking is governed by the
\emph{unrejected} component of the residual dynamics driving $\{\mathbf{e}_t\}$, i.e.\ the innovation entering the
one-step error update after conditioning on $\mathcal{F}_{t-1}$.
When we represent this unrejected forcing in standardized form, we impose the following regularity condition.

\begin{assumption}[Innovation regularity]
\label{ass:innov_reg}
The standardized innovations satisfy $\mathbb{E}[\boldsymbol{\xi}_t|\mathcal{F}_{t-1}]=\mathbf{0}$,
$\mathbb{E}[\boldsymbol{\xi}_t\boldsymbol{\xi}_t']=I_3$, and $\mathbb{E}\|\boldsymbol{\xi}_t\|^{2+\delta}<\infty$ for some $\delta>0$.
\end{assumption}

An important subtlety of this assumption deserves emphasis.

\begin{remark}[Regularity is imposed on the tracking residual, not on the environment]
\label{rem:regularity_on_e}
Equation \eqref{eq:error_shock_offset} implies that the object entering both stability and inference is the
\emph{difference process} $\mathbf{s}_t-\mathbf{o}_t$. We therefore do \emph{not} require that the disturbance
environment $\{\mathbf{s}_t\}$ be light-tailed, weakly dependent, or close to Gaussian.
Stable tracking instead requires that the guidance--navigation--control stack generate offsets $\{\mathbf{o}_t\}$
such that the \emph{net residual} driving the tracking error is regular (finite moments and stable conditional
energy), even when $\{\mathbf{s}_t\}$ itself is irregular. In the passive regime $\mathbf{o}_t\equiv \mathbf{0}$,
this reduces to a physical regularity requirement on $\{\mathbf{s}_t\}$.
\end{remark}

In the linearized recursion introduced in Section~\ref{sec:tracking}, define the one-step-ahead residual forcing as
\begin{equation}
\boldsymbol{\varepsilon}_{t+1}
\;\equiv\;
\mathbf{e}_{t+1}-A_t\mathbf{e}_t-B_t\mathbf{c}_t,
\label{eq:eps_def_from_e}
\end{equation}
where $(A_t,B_t,\mathbf{c}_t)$ are the linearized error and control-channel objects defined in
\eqref{eq:closed_loop_generic}. This $\boldsymbol{\varepsilon}_{t+1}$ is the object modeled stochastically below:
it is the \emph{unrejected} innovation in tracking error after applying the realized corrective authority.

Although $\mathbf{e}_t\in\mathbb{R}^3$ is defined in an inertial frame, it is often more interpretable in a trajectory-aligned basis. Let $\hat{\mathbf{v}}_t$ be the estimated inertial velocity. Define the local orthonormal triad
\begin{equation}
\label{eq:tnr_triad}
\mathbf{t}_t=\frac{\hat{\mathbf{v}}_t}{\|\hat{\mathbf{v}}_t\|},\qquad
\mathbf{n}_t=\frac{\mathbf{x}_t\times\hat{\mathbf{v}}_t}{\|\mathbf{x}_t\times\hat{\mathbf{v}}_t\|},\qquad
\mathbf{r}_t^{(R)}=\mathbf{n}_t\times\mathbf{t}_t,
\end{equation}
corresponding to along-track ($\mathbf{t}_t$), cross-track ($\mathbf{n}_t$), and in-plane radial ($\mathbf{r}_t^{(R)}$) directions. Projecting the tracking error onto this basis gives
\begin{equation}
\label{eq:error_components}
e_t^{(T)}=\mathbf{t}_t^\top\mathbf{e}_t,\qquad
e_t^{(R)}=\left(\mathbf{r}_t^{(R)}\right)^\top\mathbf{e}_t,\qquad
e_t^{(N)}=\mathbf{n}_t^\top\mathbf{e}_t,
\end{equation}
which represent forward/backward deviation, in-plane radial deviation, and out-of-plane deviation, respectively. This decomposition is useful because disturbances tend to load these axes differently, while active stabilization predicts error-correcting dynamics in each component.

For long-horizon navigation---where trajectories may extend over years to decades at interstellar distances~\cite{jewitt2022interstellar}---$\mathbf{x}^{\mathrm{plan}}_t$ is expected to be smooth (low curvature and sparse major retargeting), while $\mathbf{e}_t$ contains the higher-frequency imprint of disturbances and corrective action.

The control objective is \emph{stable tracking} of the reference $\mathbf{x}^{\mathrm{ref}}_t$ in the sense that the tracking error remains bounded in mean square and becomes insensitive to initialization:
\begin{equation}
\sup_t \mathbb{E}\|\mathbf{e}_t\|^2 < \infty,
\qquad
\lim_{k\to\infty}\mathbb{E}\|\mathbf{e}_{t+k}-\tilde{\mathbf{e}}_{t+k}\|^2 = 0,
\label{eq:stability_objective}
\end{equation}
where $\tilde{\mathbf{e}}_t$ denotes the error process under a different initial state estimate. These conditions formalize the requirement that stochastic forcing does not accumulate into long-run divergence from the intended path---the path that was designed in Section~\ref{sec:trajectory_geometry} to satisfy debris-avoidance and reconnaissance constraints.

The gravity term is not treated as a disturbance to be rejected; it is incorporated into the reference definition \eqref{eq:x_ref_def}. In this sense, gravity is ``canceled'' in the tracking error \eqref{eq:e_def} by construction, leaving the stability theory to focus on bounded residual dynamics via \eqref{eq:error_shock_offset}. A similar distinction applies to outgassing: any stable, slowly varying mean component can be incorporated into $\mathbf{x}^{\mathrm{plan}}_t$ at the guidance level, whereas the unpredictable component of outgassing variability enters $\mathbf{s}_t$ and must be countered by short-horizon correction. We make this separation explicit with the following assumption.

\begin{assumption}[Gravity-referenced control authority]
\label{ass:grav_ref_authority}
There exists a feasible long-horizon guidance signal $\{\mathbf{x}^{\mathrm{plan}}_t\}$---satisfying the debris-avoidance and reconnaissance constraints of Section~\ref{sec:trajectory_geometry}---and corresponding reference $\mathbf{x}^{\mathrm{ref}}_t=\mathbf{x}^{\mathrm{grav}}_t+\mathbf{x}^{\mathrm{plan}}_t$, such that the remaining short-horizon corrective authority can generate an offset process $\{\mathbf{o}_t\}$ that yields stable tracking in the sense of \eqref{eq:stability_objective}. Any temporary degradation of corrective authority (e.g.\ saturation, phase constraints, or actuator irregularities) enters the effective residual dynamics through $\mathbf{o}_t$ and can be treated as part of the realized disturbance-versus-offset balance in \eqref{eq:error_shock_offset}.
\end{assumption}

Assumption~\ref{ass:grav_ref_authority} formalizes the intended role of control authority in this paper. Long-horizon steering selects $\mathbf{x}^{\mathrm{plan}}_t$ and thus $\mathbf{x}^{\mathrm{ref}}_t$, exploiting predictable forces and any stable mean non-gravitational components to reduce propellant requirements, while short-horizon correction offsets unpredictable forcing so that the residual process $\mathbf{e}_t=\mathbf{s}_t-\mathbf{o}_t$ remains stable. This framing anticipates two complementary empirical strategies developed later: tests based on deviations from the gravity reference $\mathbf{x}^{\mathrm{grav}}_t$, and tests that seek evidence of $\mathbf{o}_t\not\equiv 0$ through error-correcting and disturbance-attenuating residual dynamics.

\subsection{Design trade-offs: shocks versus control authority}
\label{sec:tradeoffs}

Equation \eqref{eq:error_shock_offset} implies that engineering requirements are determined by the
\emph{unrejected} disturbance in tracking error, not by the raw disturbance field alone.
Section~\ref{sec:tracking} formalizes persistence and clustered energy by parameterizing the innovation channel
that drives $\{\mathbf{e}_t\}$; Assumption~\ref{ass:innov_reg} provides the baseline regularity conditions on
the standardized innovations underlying that representation.

Persistent residual disturbances tighten the required contraction of tracking dynamics. Let
\begin{equation}
\bar{\theta}_t \equiv \sum_{j=1}^q \|\Theta_{j,t}\|
\end{equation}
measure total moving-average (MA) persistence in the innovation channel, where $\Theta_{j,t}$ are the MA coefficient matrices and $q$ is the MA order (both formalized in \eqref{eq:dist_ma} below).
When residual forcing is more persistent, the control loop must attenuate errors more aggressively (through gain,
update rate, and/or authority margins) to prevent accumulation.

\begin{proposition}[AR--MA trade-off under persistent residual disturbances]
\label{prop:ar_ma_tradeoff}
Let $\mathcal{M}$ denote the companion matrix governing the closed-loop tracking-error state recursion (formalized in \eqref{eq:tv_state} below), and let $\|\cdot\|$ be any submultiplicative matrix norm. A sufficient condition for bounded mean-square tracking error under MA-persistent innovations is
\begin{equation}
\|\mathcal{M}\| < \frac{1}{1+\bar{\theta}_t},
\label{eq:ar_ma_bound}
\end{equation}
so that high residual disturbance persistence requires stronger contraction (higher feedback gain and/or higher correction frequency). This follows from the standard Lyapunov recursion for the mean-square state: submultiplicativity gives $\|\mathcal{M}^k\|\le\|\mathcal{M}\|^k$, while the MA expansion of the innovation amplifies shocks by at most $(1+\bar{\theta}_t)$; requiring the product to contract yields \eqref{eq:ar_ma_bound}~\cite{hamilton1994,lutkepohl2005}.
\end{proposition}

Proposition~\ref{prop:ar_ma_tradeoff} quantifies the engineering tension between the perturbation
environment and the required control bandwidth. Importantly, the object governed by
\eqref{eq:ar_ma_bound} is the \emph{residual} forcing that remains after feasible offsetting.
Thus, the space environment can be severe provided the craft can match it closely enough that the
innovation process driving $\mathbf{e}_t$ remains regular.

Over-actuated designs (larger $m$, the number of available jets; see Section~\ref{sec:jets}) can distribute corrective effort across jets and relax thermal and duty-cycle constraints, but the governing stability inequalities remain of the same form. Minimal configurations that achieve three-dimensional reachability with few jets have naturally tighter margins, implying that operational policies (e.g.\ duty cycles, warm bias thrust, and phase-aware correction scheduling) must be selected more conservatively to maintain tracking performance under worst-case disturbance sequences.

\section{Tracking as a stochastic dynamical system}
\label{sec:tracking}

The tracking error $\mathbf{e}_t = \mathbf{s}_t - \mathbf{o}_t$ is shaped by two distinct channels, and their relative contributions determine whether residual dynamics behave as passive or actively controlled. The \emph{disturbance channel} carries the imprint of the space environment: solar-wind persistence, episodic impulses, and variable outgassing induce temporal autocorrelation and clustered energy in $\mathbf{s}_t$. In the passive regime ($\mathbf{o}_t\equiv\mathbf{0}$), these features propagate directly into the tracking-error innovation $\boldsymbol{\varepsilon}_t$. Under active control, the offset $\mathbf{o}_t$ partially cancels disturbance structure before it reaches $\mathbf{e}_t$, so the innovation captures only the unrejected component---effective stabilization attenuates both persistence and energy clustering. However, corrective action opens a second pathway: the \emph{actuator channel}, through which thruster memory---thermal inertia, valve latency, feed-system dynamics---introduces its own serial dependence into $\mathbf{e}_t$ via the control term $B_t\mathbf{c}_t$ in the error recursion \eqref{eq:closed_loop_generic}. This dependence is ideally stable and fast-decaying, so that thruster memory does not reintroduce the persistence that corrective action removed from the disturbance channel. Formally, we parameterize the disturbance channel with a moving-average (MA) structure for persistence and a conditional-heteroskedasticity (GARCH) structure for clustered energy (\eqref{eq:dist_ma}--\eqref{eq:H_sre}), and the actuator channel as a VARMA process on $\mathbf{c}_t$ (\eqref{eq:c_varma}). The reduced-form companion representation \eqref{eq:tv_state} absorbs both channels into a single state recursion, so the multi-level stability conditions derived in Section~\ref{sec:stability} govern the combined dynamics. This two-channel structure is also what makes the residual diagnostics in Section~\ref{sec:results} informative: under comparable environmental forcing, an actively controlled body should exhibit weaker persistence, less variance clustering, and more regular innovations than a passive body, because the actuator channel has absorbed part of the disturbance structure before it becomes observable in $\mathbf{e}_t$.

The disturbance channel spans multiple regimes. Within the heliosphere ($\lesssim$100~AU), solar radiation pressure contributes a quasi-deterministic acceleration scaling as $r^{-2}$ with cross-sectional area and reflectivity~\cite{wertz1999spacecraft}. Episodic events---coronal mass ejections, solar energetic particle bursts, and stream interaction regions---add transient accelerations with heavy-tailed amplitudes~\cite{richardson2018solar,tsurutani2003extreme}. 

For volatile-rich bodies, thermally driven sublimation produces a reaction force that is predominantly anti-sunward (opposing the incident solar flux), with a smaller transverse component arising from rotation-mediated thermal lag and surface heterogeneity~\cite{mumma2011composition}. The resulting non-gravitational acceleration has a partially predictable component (absorbable into guidance) and a stochastic residual. Cometary rotation periods typically range from 5--70~hours, with most clustering between 5--20~hours~\cite{Kokotanekova2017Rotation,simon2019comet}; this rotation modulates both disturbance coupling and available control authority. Beyond the heliosphere, galactic cosmic rays and the interstellar medium introduce lower-intensity but persistent perturbations~\cite{jokipii1998particle}. 

Overall, these sources produce the two characteristic features---persistence and episodic energy clustering---that motivate the MA--GARCH parameterization of the innovation channel. We now (i)~provide a concrete mechanism by which $\mathbf{o}_t$ arises from jet-level actuation, formalizing the actuator channel, and (ii)~specify the stochastic model for the unrejected tracking-error innovation, formalizing the disturbance channel.

\subsection{Jet-based actuation: active trajectory control}
\label{sec:jets}

We now relate the corrective offset $\mathbf{o}_t$ to jet-level actuation. This link is not required for the stability framework, but it helps interpret how corrective action may appear in observables such as astrometric residuals, duty cycles, and photometric variability.

Let $m$ be the number of available jets. Each jet $i$ produces a body-frame thrust vector
$f_{i,t}\mathbf{b}_i$ with scalar magnitude $f_{i,t}\ge 0$ and fixed direction $\mathbf{b}_i\in\mathbb{R}^3$
in the body frame. Define the jet command vector
\begin{equation}
\mathbf{u}_t \equiv (f_{1,t},\dots,f_{m,t})'.
\end{equation}
Let $R_t\in SO(3)$ be the body-to-inertial rotation matrix and $M_t$ the (slowly varying) mass.
The inertial-frame corrective acceleration produced by the jets is
\begin{equation}
\mathbf{c}_t = \frac{1}{M_t} R_t \sum_{i=1}^m f_{i,t}\mathbf{b}_i,
\label{eq:c_general}
\end{equation}
which can be written compactly as $\mathbf{c}_t = G_t \mathbf{u}_t$ for an appropriate time-varying mapping $G_t$
that depends on configuration, attitude, and mass. When rotation is rapid relative to control sampling, $R_t$ induces
periodically time-varying authority even for a fixed body-frame firing pattern.

Practical three-dimensional steering requires that the \emph{net} force directions available over a control interval span
three independent inertial directions; this can be achieved either with multiple fixed jets or by combining fewer jets with
rotation-mediated authority and/or dedicated attitude control. The theory and signatures developed in this paper are not tied
to a specific minimal geometry: we write the analysis for general $m$ and allow both under-actuated (rotation-mediated) and
fully actuated (dedicated attitude jets) regimes~\cite{wie2008spacecraft,markley2014fundamentals}. The standing requirement
is simply that the closed-loop mapping from navigation information to net corrective acceleration is well defined and bounded.

\begin{assumption}[Control mapping regularity]
\label{ass:h_regularity}
At each time $t$, the control system selects jet commands $\mathbf{u}_t$ as a measurable function of an information set
$\mathcal{I}_t$ (state estimates, attitude, actuator states):
\[
\mathbf{u}_t = h(\mathcal{I}_t,\mathcal{J}),
\]
where $\mathcal{J}$ denotes the jet configuration. The mapping $h$ remains defined under sensor noise and transient radiation effects,
and satisfies a configuration-dependent bound $\|\mathbf{u}_t\|\le \bar{u}(\mathcal{J})$.
\end{assumption}

In the linearized discrete-time error dynamics below, $\mathbf{o}_t$ should be interpreted as the \emph{effective} correction
applied through the input channel. Linearizing around the reference trajectory $\mathbf{x}^{\mathrm{ref}}_t$, acceleration enters
the error update via a (possibly time-varying) matrix $B_t$, so the induced correction term is $B_t\mathbf{c}_t$. This makes the
identification explicit:
\begin{equation}
\mathbf{o}_t \equiv -\,B_t \mathbf{c}_t,
\label{eq:o_from_c}
\end{equation}
up to sign conventions in the chosen error recursion. Rotation enters through $G_t$ in \eqref{eq:c_general} and therefore induces
periodic modulation in the effective correction term $B_tG_t\mathbf{u}_t$.
The practical realization of this mapping is subject to resource constraints.

\begin{remark}[Propellant and operational constraints]
\label{rem:propellant}
For missions exploiting cometary volatiles as propellant~\cite{mumma2011composition,altwegg2019rosina}, available thrust authority depends on extraction,
processing, and feed-system constraints. Operational considerations may favor quasi-continuous low-throttle ``warm'' modes over strict on/off strategies,
which affects stability margins and observable duty cycles.
\end{remark}

\subsection{Perturbation properties in space}
\label{sec:perturbations_properties}

We now formalize the residual dynamics implied by the decomposition in Section~\ref{sec:problem} and the actuation mapping above.
The stability and inference objects are the tracking error $\mathbf{e}_t$ and its innovation process. Importantly, the stochastic
assumptions below are imposed on the \emph{unrejected tracking residual}---the portion of $\mathbf{s}_t$ not cancelled by feasible
$\mathbf{o}_t$---rather than on the physical disturbance field $\mathbf{s}_t$ in isolation.

We model the tracking error in linearized discrete-time form:
\begin{equation}
\mathbf{e}_{t+1} = A_t\mathbf{e}_t + B_t\mathbf{c}_t + \boldsymbol{\varepsilon}_{t+1},
\label{eq:closed_loop_generic}
\end{equation}
where $\boldsymbol{\varepsilon}_{t+1}$ is the tracking-error innovation, equivalently defined by \eqref{eq:eps_def_from_e}.
The matrices $A_t$ and $B_t$ may be time-varying due to changing geometry, attitude modulation, and operating regimes.

Many stabilizing controllers also introduce internal state (e.g.\ filtered errors, adaptive gain dynamics, actuator state), which
motivates a reduced-form companion representation in a state $\mathbf{z}_t$~\cite{hamilton1994,lutkepohl2005}:
\begin{equation}
\mathbf{z}_{t+1} = \mathcal{M}_t\mathbf{z}_t + \mathcal{N}_t \boldsymbol{\eta}_t.
\label{eq:tv_state}
\end{equation}
This is not a commitment to a particular controller; it provides a convenient vehicle for generic stability conditions under time
variation induced by rotation and operating regimes. This formulation also accommodates spatial feedback that may be present in the physical system but appears instantaneous at the temporal resolution of observations. When the observation frequency is coarser than the underlying spatio-temporal dynamics, spatial interactions among coupled state components can be absorbed into the time series coefficient structure by inverting the contemporaneous dependence matrix, yielding a nonlinear reduced-form whose stability is governed by the same spectral-radius conditions derived below~\cite{andree2020dynamic}. Imposing contraction on the combined coefficient matrix then implies stability of both the spatial and temporal components, so the time series stability framework adopted here remains sufficient even under instantaneous spatial coupling at the measurement cadence.

Deep-space residual disturbances are typically neither independent nor homoskedastic. We represent persistence in the innovation
$\boldsymbol{\varepsilon}_t$ via a moving-average structure,
\begin{equation}
\boldsymbol{\varepsilon}_t = \sum_{j=0}^{q} \Theta_{j,t}\boldsymbol{\eta}_{t-j},\qquad \Theta_{0,t}\equiv I_3,
\label{eq:dist_ma}
\end{equation}
and volatility clustering via a multivariate conditional covariance recursion for the primitive shock $\boldsymbol{\eta}_t$,
\begin{equation}
\begin{aligned}
\boldsymbol{\eta}_t &= H_t^{1/2}\boldsymbol{\xi}_t,\\
\mathrm{vec}(H_t)
&=\mathbf{w}_t
+\mathcal{A}_t\,\mathrm{vec}\!\left(\boldsymbol{\eta}_{t-1}\boldsymbol{\eta}_{t-1}'\right) \\
&\quad +\mathcal{B}_t\,\mathrm{vec}(H_{t-1}).
\end{aligned}
\label{eq:H_sre}
\end{equation}
This nests common multivariate GARCH forms~\cite{bougerol1992strict,straumann2006quasi} and allows time variation across operating regimes.
Together, \eqref{eq:dist_ma}--\eqref{eq:H_sre} provide a concrete parameterization of the innovation channel described abstractly in Section~\ref{sec:problem}.

\begin{assumption}[Tracking-error innovation regularity]
\label{ass:armagarch_reg}
Let $\mathcal{F}_t$ denote the filtration generated by the observation process and all realized control actions up to time $t$
(including any internal controller state). The standardized innovations satisfy
\[
\mathbb{E}[\boldsymbol{\xi}_t\mid \mathcal{F}_{t-1}]=\mathbf{0},\qquad
\mathbb{E}[\boldsymbol{\xi}_t\boldsymbol{\xi}_t']=I_3,\qquad
\mathbb{E}\|\boldsymbol{\xi}_t\|^{2+\delta}<\infty
\]
for some $\delta>0$.
\end{assumption}

Assumption~\ref{ass:armagarch_reg} is a closed-loop regularity condition on $\boldsymbol{\varepsilon}_t$ (and hence on the innovations
of $\{\mathbf{e}_t\}$); it does not restrict $\{\mathbf{s}_t\}$ or $\{\mathbf{o}_t\}$ individually. In the passive regime
$\mathbf{o}_t\equiv\mathbf{0}$, it reduces to a physical regularity condition on $\boldsymbol{\varepsilon}_t\equiv\mathbf{s}_t$.
More generally, even if $\mathbf{s}_t$ exhibits rare/extreme episodes, bursty clustering, or regime shifts, the assumption should be read
as requiring that the effective innovation channel driving tracking error still admits the (possibly time-varying) ARMA--GARCH representation
\eqref{eq:dist_ma}--\eqref{eq:H_sre} with standardized shocks having controlled moments. Operationally, this corresponds to the guidance--navigation--control
stack and available actuation authority regularizing the net residual $\mathbf{s}_t-\mathbf{o}_t$ even when $\mathbf{s}_t$ itself is irregular.
This observation motivates a concrete diagnostic strategy.

\begin{remark}[Craft as regularizers of extreme space perturbations]
\label{rem:regularization_test}
If $\mathbf{s}_t$ exhibits episodic extreme behavior that would violate Assumption~\ref{ass:armagarch_reg} for a natural body, a passive object must inherit this
irregularity in its observed innovations. By contrast, a sufficiently stabilized active craft can generate offsets $\mathbf{o}_t$ such that the resulting tracking
error $\mathbf{e}_t=\mathbf{s}_t-\mathbf{o}_t$ remains well described by a stable, possibly time-varying ARMA--GARCH innovation channel with bounded conditional energy.
This motivates a signature test based on whether filtered tracking residuals admit regular standardized innovations across identified disturbance regimes in which the underlying disturbance environment does not.
\end{remark}

\section{Multi-level stability conditions}
\label{sec:stability}

Stable tracking in ISO-class missions is inherently multi-level, because the dynamics that can destabilize a trajectory originate from distinct physical domains. The stochastic space environment (solar wind, radiation pressure, outgassing variability) drives conditional variance in the disturbance channel and can produce escalating forcing if left unchecked (Level~1). The tracking loop itself governs whether correction epochs compound or attenuate past errors (Level~2). The actuator hardware introduces its own dynamics through thermal inertia, valve latency, and feed-system transients (Level~3). The rotation of the host body modulates thrust direction periodically, creating phase-dependent variation in control authority (Level~4). Stability at each level is necessary: violating any one can yield unbounded residual energy, persistent drift, actuator-driven oscillation, or phase-dependent loss of authority. This layered requirement reflects the broader challenge of achieving reliable on-orbit behavior under nonlinear couplings and multi-field disturbances in spacecraft dynamics and control~\cite{zhang2024advances_sdc}.

\paragraph{Level 1: disturbance-energy stability (GARCH).}
We first require stability of the conditional disturbance energy in \eqref{eq:H_sre}, so that stochastic forcing in the \emph{tracking-error innovation} channel does not exhibit explosive conditional covariance dynamics. Under Assumption~\ref{ass:innov_reg}, this ensures finite second moments for the innovation channel, which is necessary for bounded mean-square tracking error.

\begin{theorem}[Volatility recursion stability]
\label{thm:garch_stability}
Consider \eqref{eq:H_sre} with constant coefficients $(\mathcal{A},\mathcal{B})$. A unique stationary solution with $\mathbb{E}\|H_t\|<\infty$ exists if and only if
\begin{equation}
\rho(\mathcal{A}\otimes\mathcal{A}+\mathcal{B}\otimes\mathcal{B})<1,
\label{eq:garch_spectral}
\end{equation}
where $\rho(\cdot)$ denotes spectral radius. Under time-varying coefficients, stability requires a negative top Lyapunov exponent for the associated matrix product~\cite{bougerol1992strict,straumann2006quasi}.
\end{theorem}

Theorem~\ref{thm:garch_stability} provides a necessary and sufficient condition ensuring that the conditional covariance of the tracking-error innovation does not grow without bound. In engineering terms, this means the energy of the disturbance channel (after feasible offsetting) must be self-regulating: large shocks cannot trigger an escalating sequence of ever-larger conditional variances. In practice, this constrains trajectory design: the mission planner selects transit corridors that keep the residual forcing environment within the self-regulating regime, avoiding regions where shock-driven variance escalation could violate \eqref{eq:garch_spectral}.

\paragraph{Level 2: outer-loop tracking contraction.}
Given that disturbance energy is bounded (Level~1), we next require that the outer-loop error dynamics contract so that the effect of initialization and past disturbances decays over time. In the reduced-form representation \eqref{eq:tv_state}, this corresponds to stability of the companion state recursion, with contraction governed by $\rho(\mathcal{M})<1$ in the time-invariant case. Given Assumption~\ref{ass:innov_reg}, geometric contraction implies a unique stationary solution with finite second moments.

\begin{theorem}[Geometric stability via spectral radius]
\label{thm:tracking_stability}
Under Assumption~\ref{ass:innov_reg}, if $\mathcal{M}_t\equiv\mathcal{M}$ is constant with $\rho(\mathcal{M})<1$, then the companion state in \eqref{eq:tv_state} has a unique stationary solution with finite second moments, and the effect of initialization decays geometrically~\cite{hamilton1994,lutkepohl2005,andree2020dynamic}.
\end{theorem}

Theorem~\ref{thm:tracking_stability} ensures that tracking errors are progressively corrected rather than persisting or amplifying: the closed-loop system is contractive, so initial estimation errors and past disturbance accumulations do not permanently bias the trajectory. The spectral radius $\rho(\mathcal{M})$ is the operative condition because it is directly computable from the eigenvalues of the companion matrix. The link to the norm-based contraction condition in Proposition~\ref{prop:ar_ma_tradeoff} is that $\rho(\mathcal{M})<1$ guarantees the existence of a submultiplicative matrix norm under which $\|\mathcal{M}\|<1$, so that the Lyapunov contraction condition is satisfied; a proof is given by~\cite{andree2022causal}. In physical terms, a tracking deviation caused by an impulsive solar-wind event, an outgassing transient, or a debris-belt encounter decays geometrically over subsequent correction epochs---each disturbance leaves a diminishing imprint on the state rather than accumulating into secular drift. The spectral radius thus governs the effective memory horizon of the control loop: values near unity permit disturbances to linger across many epochs, while small values ensure rapid attenuation. This is the formal counterpart of the operational requirement identified in Section~\ref{sec:stability_over_steering}: the control system must prevent stochastic forcing from accumulating into long-run divergence from the planned path. Concretely, this is a precision requirement on the guidance loop: if corrections consistently undershoot or overshoot due to estimation bias or actuator calibration error, the spectral radius approaches unity and residuals drift toward a random walk.

\paragraph{Level 3: inner-loop actuation stability.}
Actuators typically exhibit memory through thermal inertia, valve latency, and feed-system constraints. Under Assumption~\ref{ass:h_regularity}, the control mapping remains bounded ($\|\mathbf{u}_t\|\le\bar{u}(\mathcal{J})$), so a generic representation for corrective acceleration dynamics is a stable VARMA recursion:
\begin{equation}
\mathbf{c}_t=\sum_{i=1}^{p_c}\Phi^{(c)}_i\mathbf{c}_{t-i}+\sum_{j=0}^{q_c}\Psi^{(c)}_j\mathbf{v}_{t-j},
\label{eq:c_varma}
\end{equation}
where $\mathbf{v}_t$ is driven by outer-loop error signals. Stability requires that the corresponding AR companion matrix has spectral radius strictly below one, ensuring finite-memory actuation and preventing self-excited oscillations or limit cycles that would be observable as persistent thrust-driven jitter~\cite{wie2008spacecraft}.

This condition ensures that the actuation subsystem itself does not introduce instability. Physically, valve response times, thermal transients in thruster chambers, and feed-system pressure dynamics must all settle within each correction cycle, so that hardware dynamics do not carry over and interfere with subsequent commands. If the actuator dynamics have eigenvalues near or above unity---for instance, a sluggish valve that has not fully closed before the next correction fires, or a thermal lag that biases thrust magnitude across consecutive epochs---the correction system could amplify rather than attenuate perturbations, producing oscillatory or divergent thrust patterns visible in both astrometric residuals and photometric signatures.

\paragraph{Level 4: rotation-mediated periodic stability.}
Finally, if net authority is modulated by spin phase through the attitude-dependent mapping in \eqref{eq:c_general}, then the relevant linearized dynamics are linear periodically time-varying. Stability is governed by Floquet multipliers~\cite{bittanti2009periodic}. Let $\mathcal{P}$ denote the monodromy matrix over one rotation cycle; stability requires $|\mu_i|<1$ for all eigenvalues $\mu_i$ of $\mathcal{P}$.

For cometary bodies with rotation periods in the viable range of approximately $5$--$20$~hours~\cite{Kokotanekova2017Rotation}---with observed values for interstellar comets ranging from roughly $7$ to $17$~hours across epochs~\cite{scarmato2026rotation,delaFuenteMarcos2025gtc}---Floquet analysis must reflect the relationship between control bandwidth and rotation period, including the possibility that the rotation state evolves under solar forcing. If corrective updates occur faster than one rotation cycle, periodic modulation tends to average out; if slower, phase-dependent authority can become the binding constraint and may dominate both stability margins and the observable structure of astrometric residuals. A disturbance arriving during a low-authority phase produces error that grows until the body rotates into a favorable orientation; Floquet stability ensures that the integrated correction over one full rotation contracts errors accumulated during any such phase. The constant-coefficient conditions stated above extend naturally to the time-varying setting.

\begin{remark}[Time-varying generalization]
\label{rem:time_varying}
The spectral-radius conditions in Levels~1--3 ($\rho(\mathcal{A}\otimes\mathcal{A}+\mathcal{B}\otimes\mathcal{B})<1$, $\rho(\mathcal{M})<1$, and the actuator companion radius) are stated for constant coefficients. Under time-varying coefficients, the natural generalization replaces spectral-radius bounds with top Lyapunov exponent conditions: stability holds if the expected log contraction rate is strictly negative~\cite{bougerol1992strict}. For the multivariate conditional-covariance recursion \eqref{eq:H_sre}, the requisite almost-sure contraction results are fully established in the univariate case but remain incomplete in the general multivariate setting~\cite{straumann2006quasi, andree2020dynamic}. From a statistical standpoint, expected log contraction is a relaxation: it permits transient episodes during which the spectral radius locally exceeds unity, provided these are offset on average by contractive phases. From an engineering standpoint, however, such transient instability is unfavorable---a craft that is stable on average but periodically loses tracking tolerance exposes the mission (and any occupants) to unacceptable risk at precisely the epochs when disturbances are strongest. For the purposes of this paper, the distinction is largely interpretive: the multi-level structure, the two-channel logic, and the residual diagnostics all carry over to the time-varying case, with the spectral-radius conditions replaced by their Lyapunov-exponent counterparts. The qualitative predictions---that effective stabilization attenuates persistence and clustering relative to the passive case---do not depend on whether the stability margins are uniform or hold only in expectation.
\end{remark}

The layered stability structure identified here is not unique to interstellar objects; analogous multi-level requirements arise wherever spacecraft operate under irregular forcing with limited actuation authority.

Station-keeping near asteroids and binary systems, for instance, involves irregular gravity fields, solar radiation pressure, and constrained actuation that impose analogous layered stability conditions~\cite{woo2016control_binary,shi2020stationkeeping_binary}. The ISO-class problem differs primarily in the dominant role of hyperbolic kinematics, the extended time horizon, and the stochastic nature of the perturbation environment, but the principle that stable control requires simultaneous satisfaction of conditions at multiple dynamical levels is common to both settings.

The framework also provides a compact mathematical lens on two planetary-defense paradigms. First, it formalizes \emph{controlled deflection}: using sustained, low-level actuation (e.g., volatile-driven jets or attached propulsion) to keep a hazardous body on a planned, collision-avoiding trajectory while rejecting perturbations. Here the objective is stable tracking: the closed-loop dynamics satisfy a contraction condition (Level~2; Theorem~\ref{thm:tracking_stability}) and the residual disturbance energy remains bounded (Level~1; Theorem~\ref{thm:garch_stability}), so small errors do not accumulate into long-run drift that could return the body toward a collision course. Second, it captures \emph{impulsive deflection} (kinetic impact): a single strike that produces an abrupt $\Delta V$ and shifts the body onto a new orbit~\cite{daly2023dart,thomas2023dart,vasile2008asteroid_deflection,sanchez2009multicriteria_neo}. Here a natural body has no corrective authority ($\mathbf{o}_t\equiv\mathbf{0}$), so the intent is to deliver a one-time change large enough that subsequent perturbations do not steer it back onto the original impact path.

\section{Observable signatures and test classes}
\label{sec:results}

The stability framework developed in Sections~\ref{sec:problem}--\ref{sec:stability} yields two complementary outputs. First, it implies \emph{design constraints and trade-offs} linking residual disturbance persistence to the control authority required for stable tracking (summarized in Section~\ref{sec:tradeoffs}). Second, it yields \emph{observable predictions} for astrometric residuals and related signatures that can, in principle, distinguish passive bodies from actively stabilized craft under comparable perturbation regimes.

The starting point is the unified trajectory decomposition introduced in Section~\ref{sec:problem}:
\begin{equation}
\mathbf{x}_t
=
\mathbf{x}^{\mathrm{grav}}_t
+
\mathbf{x}^{\mathrm{plan}}_t
+
\mathbf{s}_t
-
\mathbf{o}_t .
\tag{\ref{eq:master_decomp}}
\end{equation}
The key inference question is not whether the realized path equals a gravity-only trajectory, but whether the observed deviations behave like \emph{passive drift} ($\mathbf{o}_t\equiv \mathbf{0}$) or like \emph{actively corrected residual motion} ($\mathbf{o}_t\not\equiv \mathbf{0}$). The tests below operationalize this distinction by isolating signatures of error correction and disturbance attenuation in observable residuals.

\subsection{Two empirical test classes}
\label{sec:tests}

We formalize two complementary classes of tests. Both use the observed trajectory $\mathbf{x}_t$, but they differ in what they treat as identifiable and, consequently, in their robustness to model expansion on the passive side.

\begin{definition}[Class G: gravity-reference deviation tests]
\label{def:classG}
Define the gravity-reference deviation
\begin{equation}
\Delta_t \;\equiv\; \mathbf{x}_t - \mathbf{x}^{\mathrm{grav}}_t
\;=\;
\mathbf{x}^{\mathrm{plan}}_t + \mathbf{s}_t - \mathbf{o}_t.
\label{eq:delta_def}
\end{equation}
Class~G tests use $\Delta_t$ to assess whether departures from gravity can plausibly be explained by passive non-gravitational forcing alone. In the natural-comet regime, $\mathbf{x}^{\mathrm{plan}}_t\equiv\mathbf{0}$ and $\mathbf{o}_t\equiv\mathbf{0}$, so $\Delta_t=\mathbf{s}_t$.
\end{definition}

\begin{definition}[Class O: offset-detection tests]
\label{def:classO}
Let $\mathbf{x}^{\mathrm{ref}}_t=\mathbf{x}^{\mathrm{grav}}_t+\mathbf{x}^{\mathrm{plan}}_t$ be the gravity-referenced reference trajectory and $\mathbf{e}_t=\mathbf{x}_t-\mathbf{x}^{\mathrm{ref}}_t$ the tracking error. Class~O tests seek evidence that $\mathbf{o}_t\not\equiv\mathbf{0}$ by exploiting the residual identity
\begin{equation}
\mathbf{e}_t = \mathbf{s}_t - \mathbf{o}_t,
\label{eq:classO_core}
\end{equation}
which predicts error correction and disturbance attenuation under active stabilization.
\end{definition}

The tests developed below rely on an empirical baseline that generates short-horizon forecasts and observable residuals. While Class~G tests can be framed relative to a gravity reference, our focus is on Class~O: detecting corrective action through the behavior of filtered residuals---error correction and disturbance attenuation---which is less sensitive to the unknown long-horizon target.

\paragraph{Gravity-only is an incomplete null.}
A gravity-only comparison can detect the presence of non-gravitational acceleration, but it is not, by itself, a discriminator of artificial control. The interpretation is limited because (i) the intended path is generally unknown (a controlled craft chooses $\mathbf{x}^{\mathrm{plan}}_t$), (ii) natural processes such as outgassing and radiation pressure can generate substantial departures from gravity trajectories, with coma expansion velocities typically 0.5--1~km/s~\cite{bockelee2004borrelly,jockers2011encke} and thermally driven sublimation producing measurable non-gravitational accelerations whose magnitude can rival or exceed any anomalous component~\cite{mumma2011composition}, and (iii) the relevant null is therefore ``gravity plus a stochastic non-gravitational acceleration process,'' within which both passive forcing and potential control action may appear.

Importantly, a controlled craft may operate within the range of natural outgassing magnitudes while a natural body may operate outside of the gravity-only range, making acceleration excess relative to gravity an unreliable discriminator on its own.

The deeper difficulty with Class~G tests is that each apparent anomaly can be absorbed into a richer passive model. An unusual acceleration magnitude may reflect exotic ice composition; an atypical direction, anisotropic sublimation geometry; the absence of a visible coma, low volatile content or small nucleus size. Because thermally driven sublimation produces a reaction force that is predominantly anti-sunward---opposing the incident solar flux, with a transverse component from thermal lag~\cite{mumma2011composition}---there is a directional prior that constrains naive passive explanations and sharpens Class~G tests when the observed deviation lies far from the expected anti-sunward axis. But this sharpening has limits: a sufficiently flexible passive model can always accommodate directional structure as well, so Class~G discrimination ultimately rests on a judgment about model plausibility rather than on a structural impossibility.

\paragraph{State-space baseline with latent non-gravitational acceleration.}
We therefore use a state-space baseline that explicitly includes latent non-gravitational acceleration. Let $\mathbf{a}^{(\mathrm{ng})}_t\in\mathbb{R}^3$ denote the non-gravitational acceleration at time $t$, and define the augmented filter state
\[
\mathbf{X}_t = (\mathbf{x}_t^\top,\;\dot{\mathbf{x}}_t^\top,\;\mathbf{a}^{(\mathrm{ng})}_t{}^\top)^\top \in \mathbb{R}^9,
\]
comprising position, velocity, and latent non-gravitational acceleration. The transition model combines gravitational dynamics for $(\mathbf{x}_t,\dot{\mathbf{x}}_t)$ with a stochastic evolution for $\mathbf{a}^{(\mathrm{ng})}_t$ (e.g.\ mean-reverting, persistent, or regime-switching). Extended or unscented Kalman filtering then yields one-step-ahead position forecasts $\hat{\mathbf{x}}_{t|t-1}$ (the conditional expectation of $\mathbf{x}_t$ given observations through $t-1$) and smoothed estimates of $\mathbf{a}^{(\mathrm{ng})}_t$ without assuming that $\mathbf{x}^{\mathrm{grav}}_t$ is the intended target. The resulting forecast residuals provide the input series for the Class~O tests below.

\subsection{Class O: offset-detection tests}

Class~O tests rest on a qualitatively different footing. However complex the passive forcing model, outgassing, radiation pressure, and stochastic perturbations do not produce closed-loop contraction---that is, they cannot generate error-correcting feedback toward a reference trajectory. This is a structural absence in passive dynamics, not a quantitative threshold, and the tests below are designed to detect its violation.

\paragraph{Test O1: short-horizon forecast consistency and error correction.}
Let $\hat{\mathbf{x}}_{t|t-1}$ be the one-step-ahead forecast from the nonlinear filter. Define forecast residuals
\begin{equation}
\mathbf{r}_t \equiv \mathbf{x}_t-\hat{\mathbf{x}}_{t|t-1}.
\label{eq:residual_def}
\end{equation}
For a passively perturbed object, $\mathbf{r}_t$ inherits disturbance persistence and volatility clustering. For an actively stabilized object, $\mathbf{r}_t$ should exhibit \emph{error correction}: mean reversion and bounded variance. A practical implementation is to fit
\[
\mathbf{r}_t = \Phi \mathbf{r}_{t-1} + \mathbf{u}_t,
\]
and test $H_0:\rho(\Phi)=1$ against $H_1:\rho(\Phi)<1$. This provides an operational proxy for outer-loop contraction (Theorem~\ref{thm:tracking_stability}).

\paragraph{Test O2: disturbance attenuation during independently observed events.}
Space-weather episodes are expected to increase persistence and clustering in the disturbance process~\cite{richardson2018solar,tsurutani2003extreme}. Under passive dynamics, this propagates into forecast residuals, producing detectable ARMA/GARCH structure. Under effective control, the same external disturbance should leave a weaker imprint because it is absorbed by corrective offsets before becoming observable in $\mathbf{r}_t$.

Operationally, partition the sample into ``event'' and ``quiet'' periods using an external index $W_t$ (solar wind pressure proxies, CME arrival indicators, etc.). Compare residual diagnostics across regimes:
\begin{enumerate}
\item change in residual autocorrelation (MA structure),
\item change in conditional variance clustering (ARCH/GARCH tests),
\item correlation between $\|\mathbf{r}_t\|$ and $W_t$.
\end{enumerate}
The controlled-signature prediction is attenuation: event periods increase these diagnostics less than expected under passive propagation.

\paragraph{Test O3: distributional restoration under shocks (Kolmogorov-style).}
If external perturbations have heavy tails (e.g.\ episodic particle bursts), then passive residuals should become less Gaussian during event periods. A controlled craft that cancels part of these shocks should produce \emph{more nearly Gaussian} innovations after conditioning on low-order AR structure.

A practical procedure is:
\begin{enumerate}
\item Fit ARMA dynamics to each component of $\mathbf{r}_t$ (or to $\|\mathbf{r}_t\|$) in rolling windows.
\item Standardize residuals to obtain $\hat{\xi}_t$.
\item Apply a Kolmogorov--Smirnov (or related EDF) test against the quiet-window baseline distribution and compare event versus quiet windows.
\end{enumerate}
The prediction is \emph{relative restoration}: event-window residuals remain closer to baseline than would be expected under measured disturbance intensity. This does not require perfect steering; it requires only that residuals are \emph{more regular than predicted} under comparable shocks.

\paragraph{Test O4: bounded-horizon return to forecasted manifolds.}
For large shocks, perfect cancellation is unlikely. A robust signature is therefore the presence of \emph{return dynamics}: after an identified disturbance, the trajectory returns toward a previously forecasted location or manifold within a bounded horizon.

Let $\hat{\mathbf{x}}^{\mathrm{free}}_{t+\tau|t}$ be the $\tau$-step forecast under the filter when non-gravitational acceleration is propagated passively. Measure the realized deviation at $t+\tau$:
\[
d_{t,\tau} = \|\mathbf{x}_{t+\tau}-\hat{\mathbf{x}}^{\mathrm{free}}_{t+\tau|t}\|.
\]
Under effective control, $d_{t,\tau}$ should exhibit systematic post-shock reduction for some range of $\tau$, consistent with bounded-horizon correction.

\begin{remark}[Authority dependence of recovery signatures]
The engineering tension formalized in Section~\ref{sec:tradeoffs} implies that recovery dynamics are configuration-dependent. A minimal jet system operating near its authority limits may be unable to fully reject a sufficiently strong shock, resulting in a temporary weakening of the contraction conditions in Theorem~\ref{thm:tracking_stability}. In such cases, the craft may require bounded-horizon recovery toward a forecasted manifold, rendering Test~O4 informative. By contrast, a high-authority system with substantial margin can absorb comparable shocks without exceeding its corrective capacity, maintaining $\mathbf{e}_t$ within nominal bounds throughout. For such systems, post-shock return dynamics may be difficult to distinguish from ordinary error correction, rendering Test~O4 less diagnostic.
\end{remark}

\section{Numerical illustration}
\label{sec:numerical_illustration}

We illustrate the signatures through a stylized simulation contrasting passive drift ($\mathbf{o}_t\equiv\mathbf{0}$) with active stabilization ($\mathbf{o}_t\not\equiv\mathbf{0}$) under an identical disturbance realization, isolating signatures attributable to offset control. The reference geometry reflects the coupled debris-avoidance and reconnaissance constraints derived above; the simulation focuses on whether tracking residuals exhibit the error-correcting and disturbance-attenuating dynamics implied by the stability framework.

\subsection{Trajectory design rationale}

Section~\ref{sec:trajectory_geometry} shows that opposite-side crossings of a debris belt of effective thickness $h$ over chord length $d$ require a minimum slope $\theta^*=\arctan(h/d)$ (Proposition~\ref{prop:debris_slope}), yielding a Kuiper-belt scale bound of order $\theta^*_{\mathrm{KB}}\approx 6^\circ$ under conservative thickness and traverse assumptions (Equation~\eqref{eq:kuiper_slope}). The same section further clarifies that the mission value of reducing slope inside the belt depends on node placement and in-plane alignment (Propositions~\ref{prop:range}--\ref{prop:node} and Corollary~\ref{cor:meaningful}), and culminates in the design principle that the propulsion problem is asymmetric: steering requirements are degree-scale and episodic, whereas stability requirements are continuous and mission-critical (Section~\ref{sec:stability_over_steering}).

Accordingly, the numerical illustration conditions on a planned reference trajectory $\mathbf{x}_t^{\mathrm{ref}}$ that is already consistent with belt-safe entry/exit geometry and a chosen inner-system node placement. The simulation then evaluates whether offset control produces the residual signatures predicted by the theory: bounded tracking error, attenuated persistence, and regularized innovation behavior during extreme disturbance regimes.

\subsection{Simulation setup}

Figure~\ref{fig:stability-signatures} presents an eight-panel numerical demonstration based on the following components. The parameter values are chosen to produce clear visual separation between passive and active residual dynamics with the intent to confirm the theoretical predictions of Section~\ref{sec:stability} rather than being calibrated to specific astrophysical environments and astrometric pipelines.

\begin{figure}[h!]
\centering
\includegraphics[width=\textwidth]{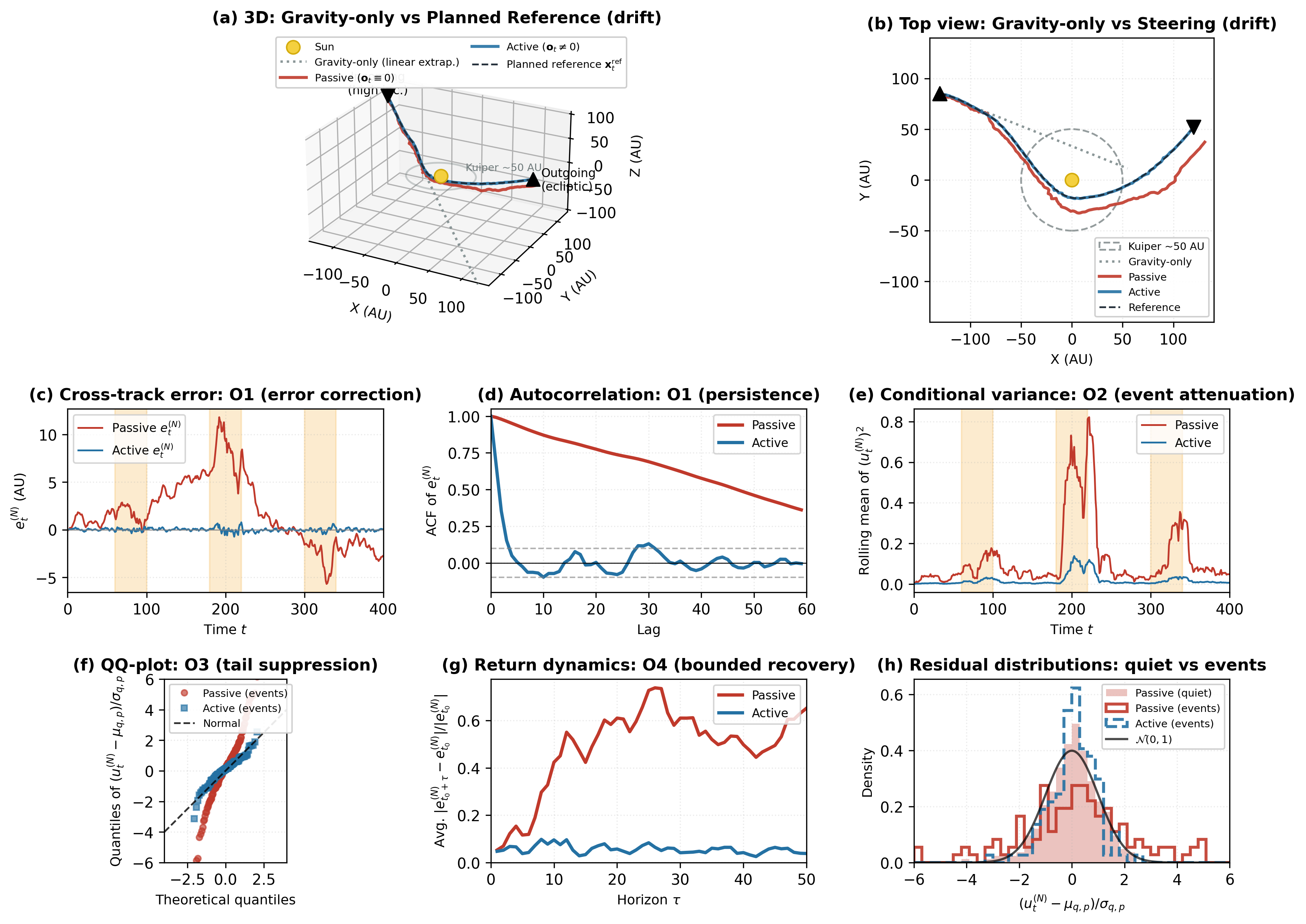}
\caption{\textbf{Stability signatures under identical stochastic forcing.} Trajectory context: (a) 3D view and (b) ecliptic projection for gravity-only baseline (gray dotted), planned reference (black dashed), passive realization (red; $\mathbf{o}_t\equiv\mathbf{0}$), and active realization (blue; offset correction). The Kuiper belt reference at 50~AU and start/end markers are shown for orientation. Signature panels map to Tests O1--O4: (c) cumulated cross-track error $e_t^{(N)}=s_t^{(N)}-o_t^{(N)}$; (d) ACF of $e_t^{(N)}$; (e) rolling mean of $(u_t^{(N)})^2$ (12-step window) with AR(1) coefficient estimated on quiet-to-quiet segments; (f) event-period QQ-plots for standardized innovations (standardized by quiet-period passive $(\mu_{q,p},\sigma_{q,p})$); (g) normalized recovery metric $\bar d_\tau$ from event-peak shock times; (h) event-period KDEs against the quiet baseline. Amber bands denote event windows (elevated disturbance variance). The reference geometry is chosen to be consistent with the debris-avoidance/node-placement setup in Section~\ref{sec:trajectory_geometry}, but the residual signatures arise from control, not from the trajectory shape.}
\label{fig:stability-signatures}
\end{figure}

\paragraph{Reference trajectory.}
The planned reference $\mathbf{x}_t^{\mathrm{ref}}$ (black dashed) is generated via shape-preserving cubic interpolation (PCHIP) through four art-directed waypoints selected to encode the geometry of Section~\ref{sec:trajectory_geometry}:
\begin{enumerate}
    \item[\textup{(i)}] inbound approach at large heliocentric distance, entering above the reference plane with belt-safe margin;
    \item[\textup{(ii)}] interior node placement and perihelion geometry (including a gravity-favorable passage consistent with the planned mission manifold);
    \item[\textup{(iii)}] outbound inner-system traverse with near-plane geometry to support reconnaissance of near-ecliptic targets;
    \item[\textup{(iv)}] outbound exit below the reference plane with a residual inclination of order the Kuiper-belt avoidance scale (cf.\ Equation~\eqref{eq:kuiper_slope}).
\end{enumerate}
The specific waypoint values are illustrative; what matters is that the resulting reference path enforces opposite-side belt crossings with margin and yields a smooth long-horizon guidance component, so that the simulation isolates short-horizon residual dynamics and correction.

\paragraph{Gravity-only baseline.}
The gravity-only proxy (gray dotted) is plotted as a local linear extrapolation from the initial segment of $\mathbf{x}_t^{\mathrm{ref}}$ to illustrate the difficulty inherent in Class~G tests: in practice, neither the gravity-only baseline nor the passive trajectory is directly observed, and a mission planner could set the target trajectory equal to the gravitational one. The visual separation between curves therefore depends on unobserved counterfactuals, underscoring the challenge that gravity-reference comparisons face and motivating the Class~O diagnostics developed above.

\paragraph{Disturbance model.}
Stochastic perturbations follow a GARCH(1,1)--MA(2) process parameterized to exhibit both persistence and volatility clustering:
\begin{itemize}
    \item MA coefficients $\boldsymbol{\theta} = (1.0, 0.6, 0.25)'$, generating short-horizon autocorrelation;
    \item GARCH parameters $\omega = 0.01$, $\alpha = 0.10$, $\beta = 0.84$, producing conditional heteroskedasticity.
\end{itemize}
During event windows (amber bands in panels~c and~e of Figure~\ref{fig:stability-signatures}), the conditional variance multiplier increases eight-fold, the ARCH coefficient $\alpha$ rises by $0.06$, and a small bias term is added to induce asymmetric impulses. This stylized regime-shift construction is meant to mimic the clustered-energy structure. Critically, the \emph{same disturbance realization} is applied to both passive and active trajectories, ensuring that differences arise solely from the presence or absence of offset control.

\paragraph{Control model.}
The active trajectory implements first-order offset correction in the trajectory-aligned basis $\{T,R,N\}$:
\begin{equation}
    o_t^{(j)} \;\leftarrow\; o_t^{(j)} + \gamma_j \bigl(s_t^{(j)} - o_t^{(j)}\bigr),
    \quad j \in \{T, R, N\},
    \label{eq:control_update}
\end{equation}
with gain parameters $\gamma_N = 0.60$, $\gamma_R = 0.45$, $\gamma_T = 0.16$, reflecting prioritized cross-track authority. This update corresponds to exponential smoothing of the disturbance accumulation toward zero tracking error, consistent with the offset interpretation in Section~\ref{sec:problem} and the closed-loop residual logic developed in Section~\ref{sec:tracking}.

\subsection{Class O test results}

Figure~\ref{fig:stability-signatures} summarizes the four Class~O diagnostics (O1--O4) on a single simulated realization, contrasting passive drift ($\mathbf{o}_t\equiv\mathbf{0}$) with active offset correction. The intent is illustrative: the disturbance process is held fixed across cases, so differences reflect only the presence or absence of control.

\paragraph{Test O1: Error correction (panels c, d).}
Panel~(c) contrasts accumulation versus boundedness: passive errors grow in a near-random-walk fashion, whereas active errors remain bounded and fluctuate about the reference. Panel~(d) shows the corresponding persistence reduction: the passive ACF decays slowly, while the active ACF drops rapidly, consistent with the error-correcting dynamics implied by Theorem~\ref{thm:tracking_stability}.

\paragraph{Test O2: Disturbance attenuation (panel e).}
Panel~(e) tracks innovation energy through time. During event windows, passive $(u_t^{(N)})^2$ exhibits pronounced spikes, while the active case is attenuated, consistent with disturbance rejection in the residual channel. To limit event-induced bias, the AR(1) coefficient is estimated using quiet-to-quiet transitions as described in Section~\ref{sec:tests}.

\paragraph{Test O3: Distributional restoration (panels f, h).}
Panels~(f) and~(h) assess whether event-period standardized innovations revert toward the quiet baseline. Relative to the passive case, the active case shows reduced tail heaviness in the event QQ-plot and a distribution (KDE) that is compressed toward the quiet-period reference.

\paragraph{Test O4: Bounded recovery (panel g).}
Panel~(g) reports the average post-shock return profile,
\[
\bar{d}_\tau \;=\; \frac{1}{|\mathcal{T}|}\sum_{t_0 \in \mathcal{T}}
\frac{|e_{t_0+\tau}^{(N)} - e_{t_0}^{(N)}|}{|e_{t_0}^{(N)}|},
\]
with $\mathcal{T}$ indexing event-peak shock times. In this realization, the passive case exhibits sustained drift away from pre-shock error levels, while the active case shows bounded recovery toward the reference manifold.

Overall, the panels align with the qualitative predictions of O1--O4: active control induces faster error correction, attenuates event-conditioned variance, suppresses heavy tails in standardized residuals, and yields bounded post-shock recovery. This numerical example is not meant to establish empirical detectability under realistic observing conditions; practical inference will depend on disturbance magnitudes, cadence, and astrometric signal-to-noise (Section~\ref{sec:conclusion}).

\section{Discussion and Conclusion}
\label{sec:conclusion}

The main contribution of this work is to make the link between control-theoretic requirements and observable residual structure explicit. The proposed tests are best viewed as \emph{probabilistic discriminants} designed to reduce reliance on assumptions about long-horizon intent by focusing on residual properties implied by active stabilization. Bridging theory to practice will require end-to-end simulation studies with realistic astrometric noise and estimation pipelines, archival analyses of well-observed comets/asteroids to establish passive null statistics, sensitivity mapping across object and observing parameters, and joint observation--control models quantifying how filtering reshapes residual signatures. The tests proposed here should be read as theoretically grounded hypotheses that may inform the development of such operational tools.

\subsection{Residual interpretation}

A useful organizing view is to separate gravity-dominated motion from any feasible guidance, and treat the observable residual as the balance between disturbances and corrective action. For a passive body, deviations from gravity primarily reflect uncontrolled non-gravitational forcing. For an actively stabilized craft, the residual is the remainder after feasible compensation and is therefore expected (statistically) to exhibit stronger short-horizon error correction and weaker persistence \emph{and} variance clustering under comparable disturbance environments. Although minimal jet configurations clarify the irreducible control geometry, the stability and signature logic is configuration-agnostic: increasing $m$ expands authority margins and distributes load, but does not change the underlying multi-level requirements for bounded residual dynamics.

\subsection{Observational feasibility}

A substantial gap remains between the mathematical framework and what can be implemented with current observational systems. Residual-based inference requires sustained, well-sampled astrometry over multiple rotation cycles, ideally at cadence comparable to or shorter than the rotation period. For ISOs with rotation periods of 5--20~hours~\cite{Kokotanekova2017Rotation}---and potentially shorter precession periods under non-principal-axis rotation~\cite{scarmato2026wobble}---and observing windows limited by hyperbolic passage~\cite{jewitt2022interstellar}, this is challenging, though coordinated multi-site campaigns and space-based assets can extend coverage~\cite{engelhardt2017population,eldadi2025}. Several constraints are likely binding. Astrometric residuals are outputs of filtering pipelines that can attenuate high-frequency structure---often where corrections are most visible---potentially confounding control signatures with estimation artifacts. Detectability depends on whether environmental forcing produces astrometrically resolvable residual structure and whether any attenuation under control is large relative to measurement noise; the signal-to-noise required to reliably identify time-series structure (e.g., short-memory vs.\ clustering behavior) in typical campaigns is not yet well characterized. Small bodies have larger area-to-mass ratios and stronger non-gravitational signatures but are fainter and harder to track at high cadence, while larger bodies are easier to track but may be intrinsically less perturbed. Phase-dependent authority---a central Floquet prediction---requires reliable rotation-state estimates from photometry, and those uncertainties propagate directly into phase-resolved tests.

No single signature is decisive. Natural processes (e.g., outgassing variability, radiation pressure, Yarkovsky-type effects) can produce substantial deviations and complex residual structure, and observational cadence and noise can mimic or erase the patterns targeted. The most defensible inference comes from \emph{bundles} of mutually reinforcing evidence interpreted against a realistic passive baseline. The strongest evidence would arise from co-occurring signatures across independent measurement channels, such as rotation-locked photometric modulation aligning with epochs of reduced astrometric residual variance; thermal emission consistent with non-solar power operation at large heliocentric distances~\cite{gibson2018krusty} coinciding with trajectory correction; or spectroscopic indicators of selective volatile depletion~\cite{mumma2011composition,altwegg2019rosina} correlating with inferred propellant expenditure. The evidentiary weight comes from joint probability under physically grounded models~\cite{eldadi2025,wright2018technosignatures}, rather than from any single feature.

\subsection{Concluding remarks}

Interstellar objects motivate an unusual control problem: trajectories are geometrically constrained at the planning stage, yet the dominant operational burden is long-horizon stability under non-gravitational forcing. The analysis here makes that asymmetry explicit. For ISO-class transits, debris-avoidance geometry and encounter phasing couple into a narrow feasible set, so once a profile is selected the mission risk is governed less by episodic steering authority than by the ability to maintain tight tracking tolerances over long durations.

Within a gravity-referenced formulation, we modeled the observable tracking residual as the net balance between stochastic disturbances and corrective offsets, and derived a layered set of stability requirements that must hold simultaneously for sustained navigation. The key implication is structural rather than parametric: when stabilization is effective, the residual dynamics should exhibit contraction and regularization relative to uncontrolled natural dynamics under comparable forcing, even when the underlying environment is persistent and intermittently extreme. This provides a principled alternative to gravity-only comparisons, which can register non-gravitational acceleration but are not, by themselves, diagnostic when natural outgassing and radiation-pressure variability can generate comparable departures.

On this basis, we outlined residual-based diagnostics that are compatible with empirical testing: short-horizon error correction, attenuation of event-conditioned persistence and variance clustering during independently identified disturbance episodes, more regular standardized innovations, and bounded recovery after shocks. A stylized numerical illustration shows how these signatures can arise from offset control under a shared disturbance realization, while underscoring that practical detectability will depend on cadence, astrometric precision, and the fidelity of end-to-end state-estimation pipelines.

The temporal contraction conditions derived here presuppose that the underlying geometric steering authority---thruster placement ensuring full directional coverage and appropriate allocation of control bandwidths---is already in place. Such spatial feedback stability is a necessary prerequisite for the tracking problem studied here, but it is not sufficient on its own: without the temporal contraction conditions, geometric authority alone cannot prevent stochastic disturbances from accumulating into long-run trajectory divergence.

Because discovery opportunities are rare and observing windows short, the framework is intended to support structured anomaly screening and to help prioritize coordinated follow-up when future interstellar visitors provide limited but high-value measurement opportunities~\cite{eldadi2025}. More broadly, it links engineering feasibility to observational consequences through stability theory, and thereby clarifies what would constitute discriminating evidence in the next ISO encounter.

\section*{Acknowledgments}

The findings, interpretations, and conclusions expressed herein are entirely those of the author and do not necessarily represent the views of the International Bank for Reconstruction and Development/World Bank, its Board of Executive Directors, or the governments they represent. No World Bank resources were used to conduct this research. This research received no specific funding. The author thanks Prof.\ Avi Loeb and acknowledges Dr.\ Frank Laukien for insights that informed this work.

\subsection*{Author Contributions}
B.~P.~J.~Andr\'ee conceived the study, developed the theoretical framework, designed and implemented the numerical simulations, and wrote the manuscript.

\subsection*{Funding}
This research received no specific funding.

\subsection*{Conflicts of Interest}
The author declares that there is no conflict of interest regarding the publication of this article.

\section*{Supplementary Materials}

\noindent Accompanying work~\cite{andree2026thruster} develops the minimal actuation architecture and its operating regimes.

\subsection*{Code Availability}
Reproducible code for the numerical illustrations is available as part of the supplementary materials. No experimental data were generated in this study.
Code and data to reproduce results and figures are archived on Zenodo at \url{https://doi.org/10.5281/zenodo.18987248}.

Code S1. Reproducible simulation code for the numerical illustration (Section~\ref{sec:numerical_illustration}).

\printbibliography

\end{document}